\documentclass[conference]{IEEEtran}

\usepackage{amsmath, amssymb, amsbsy, cite}
\usepackage{mathtools}
\usepackage[small]{caption}
\usepackage{amsthm,multirow,color,amsfonts}
\usepackage{tabulary}
\usepackage{subfigure}
\usepackage{graphicx}
\usepackage{setspace}
\usepackage{enumerate}
\usepackage{bbm,bm}
\usepackage[]{algorithm2e}
\usepackage{comment}
\usepackage{setspace}	
\usepackage{enumitem}
\makeatletter
\def\BState{\State\hskip-\ALG@thistlm}
\makeatother
\newtheorem{prop}{Proposition}

\newtheorem{cor}{Corollary}

\newcommand{\matr}[1]{\mathbf{#1}} 

\setlist[itemize]{leftmargin=*}

\title{ARQ with Cumulative Feedback to Compensate for Burst Errors}

\author{
  \IEEEauthorblockN{Derya Malak and Muriel M\'{e}dard}
  \IEEEauthorblockA{RLE, Massachusetts Institute of Technology\\
                    Cambridge, MA 02139 USA\\
                    Email: \{deryam, medard\}@mit.edu}
  \and
  \IEEEauthorblockN{Edmund M. Yeh} 
  \IEEEauthorblockA{ECE Dept., Northeastern University\\ 
                    Boston, MA 02115, USA\\
                    Email: eyeh@ece.neu.edu}
    }

\IEEEoverridecommandlockouts
\begin{document}

\maketitle
\begin{abstract}
We propose a cumulative feedback-based ARQ (CF ARQ) protocol for a sliding window of size $2$ over packet erasure channels with unreliable feedback. We exploit a matrix signal-flow graph approach to analyze probability-generating functions of transmission and delay times. Contrasting its performance with that of the uncoded baseline scheme for ARQ, developed by Ausavapattanakun and Nosratinia, we demonstrate that CF ARQ can provide significantly less average delay under bursty feedback, and gains up to about 20\% in terms of throughput. We also outline the benefits of CF ARQ under burst errors and asymmetric channel conditions. The protocol is more predictable across statistics, hence is more stable. This can help design robust systems when feedback is unreliable. This feature may be preferable for meeting the strict end-to-end latency and reliability requirements of future use cases of ultra-reliable low-latency communications in 5G, such as mission-critical communications and industrial control for critical control messaging.
\end{abstract}

\maketitle


\section{Introduction}
\label{intro}
Ultra reliability and low latency in 5G are key factors for many applications ranging from industrial automation, tactile Internet, remote healthcare, public safety, to mission-critical communications such as autonomous driving and wearable computing devices \cite{Popetal2018,Fettweis2014,3GPP2018Mar}. 5G will need to support a round-trip time (RTT) of about 1 millisecond, along with necessary overheads for resource allocation and access in 5G networks. Such severe latency constraints introduce a plethora of challenges in terms of the protocol stack design, control/user plane, and the core network \cite{Andrews2014}. 

Repetition of a packet over non-deterministic channel conditions, and the use of forward error correction (FEC) codes help repair the loss of the packets. Feedback packets are used to request FEC retransmission for increasing the reliability in packet delivery. The role of feedback is to increase data channel efficiency by limiting the repetitions. However, coding and feedback have been difficult to blend. 

Reliable communication over a packet erasure channel can be achieved using Automatic Repeat reQuest (ARQ), when there is full feedback \cite{SunShaMed2008}. This simple scheme achieves 100\% throughput, in-order delivery and the lowest possible packet delay, and it is composable across links. However, when the network is lossy, i.e., with no idealized feedback, link-by-link ARQ cannot achieve the capacity of a general network. 

In the literature, the achievable rate has been optimized using acknowledgments and coding, under the condition that each received packet is either useless or can be immediately decoded by the destination \cite{KatRahHuKatMedCrow2006}. Feedback and coding over a broadcast erasure channel have been combined in \cite{KelDriFra2008} to optimize decoding delay when perfect feedback is available from the receivers. An extension of ARQ for coded networks has been proposed in \cite{SunShaMed2008} to minimize the queue size at the transmitter. This approach combines the benefits of network coding and ARQ by acknowledging degrees of freedom (DoF) instead of original packets. It enables the feedback-based control of the tradeoff between throughput and decoding delay \cite{FraLunMedPak2007}. The proposed scheme in \cite{SunShaMed2008} is robust to delayed or imperfect feedback. None of these examples jointly investigate the delay and throughput when the feedback is imperfect. 

For schemes requiring feedback, it is generally assumed that feedback is lossless (perfect) and instantaneous (delay-free) \cite{SunShaMed2008}, \cite{FraLunMedPak2007}, \cite{PakFraSho2005}. Inevitable feedback channel impairments may cause unreliability in packet delivery. Burst errors might occur, which can impede the stability. The situation becomes worse under round-trip time (RTT) fluctuations along with the delayed feedback. To the best of our knowledge, the effect of unreliable feedback has not been captured before. 

In this paper, we investigate the effect of unreliable feedback in packet erasure channels. Erasure errors can occur in both the forward and reverse channels. However, an acknowledgment (ACK) cannot be decoded as a negative acknowledgment (NACK), and vice versa. Building on the uncoded baseline scheme proposed in \cite{AusNos2007}, we propose a SR ARQ scheme under a cumulative feedback-based ARQ (CF ARQ) scheme in order to investigate the role of feedback. We investigate how much we can gain with cumulative feedback and how to compensate the forward errors with cumulative feedback. Contrasting the throughput and delay performance of CF ARQ with the uncoded ARQ in \cite{AusNos2007}, we demonstrate that with a sliding window of size 2, CF ARQ can provide gains up to 18\% in terms of throughput. Cumulative feedback also has benefits under burst errors or high erasure rates.

\section{Channel Model} 
\label{model}
We have a point-to-point channel model consisting of a sender and a receiver. In the forward link, the sender attempts to transmit a packet to the receiver, and upon the successful reception of the packet, in the reverse link, the receiver acknowledges the sender by transmitting a feedback. We use a Gilbert-Elliott (GE) model\footnote{A general finite-state Markov model can be used to represent a physical channel with fading. The received signal-to-noise ratio can be partitioned into a finite number of states, corresponding to different channel qualities \cite{ZhaKas1999}.} \cite{Elliott1963}, which is a special case of hidden Markov models (HMMs), both for the forward and reverse channels. The status of a transmission at time $t$ is a random variable taking values in $\mathcal{X}=\{0,1\}$, where $0$ denotes an error-free packet, and $1$ means the packet is erroneous. This binary-state Markov process $S_t$, with probability transition matrix $\matr{P}$, has states G (good) and B (bad), i.e. $\mathcal{S}=\{G,B\}$, with $\bm{\epsilon}=[\epsilon_G, \epsilon_B]$ where $\epsilon_G$ and $\epsilon_B$ are the probabilities of transmitting a packet in error in the respective states. The GE channel $X_t$, driven by $S_t$, is characterized by $\{\mathcal{S},\mathcal{X},\matr{P},\bm{\epsilon}\}$. 

The channel state information is not available at the transmitter and the receiver. Hence, the transmitter does not know the status of a transmission (state of the forward link) at time $t$, but it observes the status of the feedback at time $t-1$, which is a Bernoulli random variable taking values in $\mathcal{X} = \{0,1\}$. Similarly, the receiver does not know the status of the reverse link, but it observes the status of a transmission at time $t$, which is a Bernoulli random variable taking values in $\mathcal{X} = \{0,1\}$. The transmitter and receiver do not observe the process $\mathcal{X}^{(c)}$. However, for the GE channel, given the channel state at time $t-1$, the joint probabilities of channel state and observation at time $t$ can be computed using the state-transition probabilities. For a GE channel, the state-transition matrix is 
\begin{align} 
\label{TransMatrix}
\matr{P}=
\begin{bmatrix}
    1-q & q  \\
    r & 1-r
\end{bmatrix},
\end{align}
where the first and second rows correspond to states G and B. The erasure rate is $\epsilon=\pi\bm{\epsilon}^\intercal$, where $\pi=[\pi_G,\,\, \pi_B]$ is the stationary vector of $\matr{P}$, which is found by solving $\pi \matr{P} = \pi$ and $\pi \matr{1} = 1$. Note that $1/r$ represents the average error burst. Hence, burst errors occur when $r$ is low. 

The joint probabilities of channel state and observation at time $t$, given the channel state at time $t-1$, are given as
\begin{align}
\mathbb{P}(S_t=j, X_t=1 \vert S_{t-1}=i)=p_{ij}\epsilon_j,\nonumber
\end{align}
which can be collected into a matrix of transition probabilities $\matr{P}_1=\matr{P}\cdot {\rm diag}\{\bm{\epsilon}\}$. Similarly, define $\matr{P}_0=\matr{P}\cdot {\rm diag}\{\matr{1}-\bm{\epsilon}\}$. The entries in matrices $\matr{P}_0$ and $\matr{P}_1$ are state-transition probabilities when viewed jointly with the conditional channel observations \cite{AusNos2007}. Hence, the HMM is characterized by $\{\mathcal{S},\mathcal{X},\matr{P}_0,\matr{P}_1\}$. 

In practice, data packets and acknowledgments typically have different lengths and different coding levels. Therefore, the erasure rates $\bm{\epsilon}$ and the parameters $r$ and $q$ of the forward and reverse channels are not necessarily the same, which is accounted in our model. Denote by $\matr{P}^{(f)}$ and $\matr{P}^{(r)}$ the state-transition matrices for the forward and reverse channels, respectively. The forward link $\{\mathcal{S}^{(f)},\mathcal{X}^{(f)},\matr{P}_0^{(f)},\matr{P}_1^{(f)}\}$ and the reverse link $\{\mathcal{S}^{(r)},\mathcal{X}^{(r)},\matr{P}_0^{(r)},\matr{P}_1^{(r)}\}$ are mutually independent. 
 
For the GE channel, the probability matrices for the forward and reverse channels $\matr{P}^{(f)}_0$ and $\matr{P}^{(r)}_1$ are given as 
\begin{align}
\matr{P}^{(f)}_0&=\matr{P}^{(f)}\cdot {\rm diag}\{\matr{1}-\bm{\epsilon}^{(f)}\} 
=
\begin{bmatrix}
    \bar{q}^{(f)}\bar{\epsilon}^{(f)}_G & q^{(f)}\bar{\epsilon}^{(f)}_B  \\
    r^{(f)}\bar{\epsilon}^{(f)}_G & \bar{r}^{(f)}\bar{\epsilon}^{(f)}_B
\end{bmatrix},\nonumber
\end{align}
\begin{align}
\matr{P}^{(r)}_1&=\matr{P}^{(r)}\cdot {\rm diag}\{\bm{\epsilon}^{(r)}\} 
=
\begin{bmatrix}
    \bar{q}^{(r)}\epsilon^{(r)}_G & q^{(r)}\epsilon^{(r)}_B  \\
    r^{(r)}\epsilon^{(r)}_G & \bar{r}^{(r)}\epsilon^{(r)}_B
\end{bmatrix}\nonumber
\end{align} 
using the shorthand notation $\bar{q} = 1 - q$, $\bar{r} = 1 - r$, $\bar{\epsilon}_G=1-\epsilon_G$ and $\bar{\epsilon}_B=1-\epsilon_B$. We can similarly compute $\matr{P}^{(f)}_1$ and $\matr{P}^{(r)}_0$. 

The composite channel is characterized by $\{\mathcal{S}^{(c)},\mathcal{X}^{(c)},\matr{P}_{00}^{(c)},\matr{P}_{01}^{(c)},\matr{P}_{10}^{(c)},\matr{P}_{11}^{(c)}\}$, where $\mathcal{S}^{(c)}=\mathcal{S}^{(f)}\times \mathcal{S}^{(r)}$ are the composite channel states, i.e. the Cartesian product of forward and reverse states, and $\mathcal{X}^{(c)}=\mathcal{X}^{(f)}\times \mathcal{X}^{(r)}=\{00,01,10,11\}$ is the combined observation set. For example, $X_t^{(c)}=10$ means the forward channel is erroneous and the reverse channel is good. For $X_t^{(c)}=11$, the joint probability of the combined observation and the composite state at time $t$, given the composite state at time $t-1$, is $(p_{ij}^{(f)}\epsilon_j^{(f)})\cdot (p_{km}^{(r)}\epsilon_m^{(r)})$. In compact notation, we have $\matr{P}_{ij}^{(c)}=\matr{P}_{i}^{(f)} \otimes \matr{P}_{j}^{(r)}$ for $X_t^{(c)}=ij$, where $\otimes$ is the Kronecker product of matrices and $i,j=0,1$. For the GE channel, the combined observation probabilities are given by the following $4\times 4$ matrices: $\matr{P}_{00}^{(c)}=\matr{P}^{(f)}_0 \otimes \matr{P}^{(r)}_0$, $\matr{P}_{01}^{(c)}=\matr{P}^{(f)}_0 \otimes \matr{P}^{(r)}_1$, $\matr{P}_{10}^{(c)}=\matr{P}^{(f)}_1 \otimes \matr{P}^{(r)}_0$, and $\matr{P}_{11}^{(c)}=\matr{P}^{(f)}_1 \otimes \matr{P}^{(r)}_1$. The combined state-transition matrix for the GE channel, i.e., $\matr{P}^{(c)}$, is a $4\times 4$ matrix that is given by the Kronecker product of $\matr{P}^{(f)}$ and $\matr{P}^{(r)}$, i.e. $\matr{P}^{(c)}=\matr{P}^{(f)} \otimes \matr{P}^{(r)}$. 

In the rest of the paper, we will drop the superscript $^{(c)}$ and denote the $4\times 4$ observation probability matrices by $\matr{P}_{00}$, $\matr{P}_{01}$, $\matr{P}_{10}$ and $\matr{P}_{11}$. We also let $\matr{P}_{0x}=\matr{P}_{00}+\matr{P}_{01}$ and $\matr{P}_{1x}=\matr{P}_{10}+\matr{P}_{11}$ be the probability matrices of success and error in the forward channel, respectively, and let $\matr{P}_{x0}=\matr{P}_{00}+\matr{P}_{10}$ and $\matr{P}_{x1}=\matr{P}_{01}+\matr{P}_{11}$ be the matrices of success and error in the reverse channel, respectively.  Furthermore, we let the matrices $\matr{P}$, $\matr{P}_0$, $\matr{P}_1$ denote the $4\times 4$ composite channel matrices.

\section{ARQ with Cumulative Feedback}
\label{CFcoding}
We propose a cumulative feedback-based ARQ (CF ARQ) scheme with coding for data transmission. It is an extension of the slotted SR ARQ, which allows the receiver to accept packets out of order, which can be stored in a buffer and sorted at the receiver to ensure in-order final delivery. Assume that all packets are available at the sender prior to transmission, the receiver does not have buffer overflows, and there is a synchronous transmission from the sender to the receiver.

We consider minimum coding, i.e., with a sliding window of size $M=2$. The protocol can easily be generalized to packet streams with $M>2$, which is out of the scope of the current paper. This scheme differs from the uncoded ARQ in \cite{AusNos2007} in the sense that the transmitted packet stream is MDS coded, and the feedback is cumulative for $M=2$ coded packets. However, the transmission scheme is repetition-based, i.e. the transmission rate is not adjusted based on the cumulative feedback. The receiver needs both coded packets to reconstruct the transmitted packet stream, i.e., the degrees of freedom (DoF) required at the receiver is $N=2$. We do not assume in-order packet delivery. Hence, the transmitted stream will be successfully decoded when both of the coded packets are successfully received and acknowledged by the receiver. 

The feedback, i.e. ACK and NACK messages sent by the receiver indicating if it has correctly received a data packet, acknowledges all correctly received packets, and is cumulative for $M=2$ coded packets. After the start of transmission (I), it takes $k-1$ time slots between the transmission of the second packet and receipt of its feedback. Therefore, the round-trip time (RTT) of CF ARQ is ${\rm RTT}=k+M-1=k+1$ slots. If the feedback was not cumulative, i.e., the first feedback was received $k-1$ slots after the transmission of the first packet, then the RTT would have been $k$ slots. A timeout mechanism is used at the transmitter to achieve reliable data transmission. When a packet stream is (re)transmitted, the timeout is set to $T$ that is greater than RTT. If the sender does not receive an acknowledgment before the timeout, it retransmits the packets until it receives an acknowledgment. Hence, we do not have an upper bound on the maximum number of retransmissions. 

The ACK/NACK sent in each slot. The packet whose ACK is lost will be acknowledged by subsequent ACKs/NACKs. If the succeeding ACKs/NACKs are successfully received before timer expiration, the packet will not be retransmitted. If the timeout expires and no ACK is received, the packet will be retransmitted. When a packet is lost and its NACK is received, the packet will be retransmitted immediately. If the NACK is also lost, the packet will be retransmitted after the timer expires. The transmission protocol is illustrated in Fig. \ref{protocol}. 

\begin{figure}[t!]
\centering
\includegraphics[width=\columnwidth]{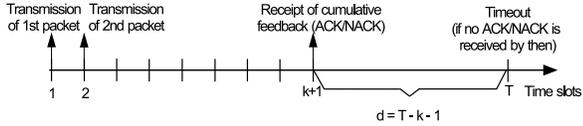}
\caption{\small{CF ARQ protocol description.}\label{protocol}}
\end{figure} 

The combined observation set for CF ARQ with $M=2$ packets is all 3-tuples of $\mathbb{Z}_2=\{0,1\}$, i.e., $\mathcal{X}^{(c)}=\mathbb{Z}_2^3$. For example, $X_t^{(c)}=001$ means that the forward channel is good for both packets and the reverse channel is erroneous, i.e., the ACK for both packets is lost at time $t$. Since the feedback is cumulative for $M=2$ packets, it is possible that both packets are successfully acknowledged, or they both need to be transmitted or only one of the packets has to be retransmitted. 

Hidden Markov model (HMM) is a statistical Markov process with unobserved states. Although the state is not directly observed, the output dependent on the state can be observed. Thus, under unreliable channel conditions, the analysis of ARQ protocol is possible using HMMs. The analysis of finite-state HMMs can be streamlined using flow graphs. Scalar-flow graphs have been used to find the probability-generating functions (PGFs) of transmission and delay times \cite{LuChang1993,ChoUn1994,LuChang1989}.  

HMMs can be analyzed by labeling the branches of scalar-flow graphs with observation probability matrices. The nodes of the flow graphs correspond to the states of the transmitter. The input node ($I$) represents the start of transmission, and the output node ($O$) represents correct reception of acknowledgment. Other nodes represent intermediate states. Upon the start of transmission, the transmitter goes from one state to the other. A state transition is accompanied with a certain value for the random variable $X$, and a probability $p$, which together appear in the branch gain $pz^X$. Hence, the input-output gain of the graph is a polynomial in $z$, whose coefficients are the probabilities of corresponding values of $X$. This polynomial is equivalent to $\mathbb{E}[z^X]$, the PGF for $X$. Flow graphs with matrix branch transmissions and vector node values are called matrix signal-flow graphs (MSFGs) \cite{AusNos2007}. The matrix gain of the graph is calculated using the basic equivalences known as parallel, series, and self-loop. The matrix-generating function (MGF) $\mathbf{\Phi}(z)$ gives the input-output relationship for the matrix-flow graph. Then, the PGF is calculated by pre- and postmultiplications of row and column vectors, respectively.

\begin{figure}[t!]
\centering
\includegraphics[width=\columnwidth]{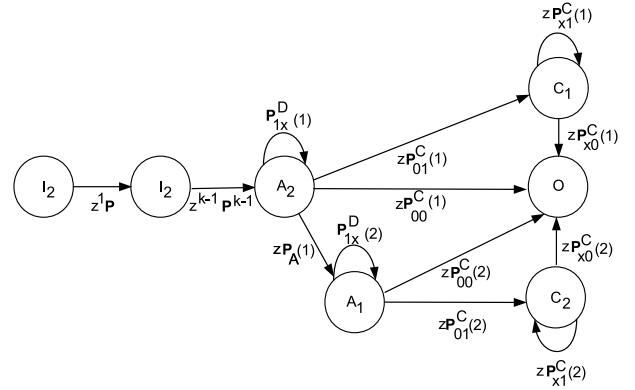}
\caption{\small{Matrix-flow graph for delay of CF ARQ in imperfect feedback.}
\label{SR_ARQ_CF}}
\end{figure}

The HMM for delay analysis of CF ARQ is shown in Fig. \ref{SR_ARQ_CF}.  The states $I$ and $O$ are the input and output nodes, and nodes $A_1$, $A_2$, $C_1$, $C_2$ represent the hidden states. The possibilities upon the transmission of $M=2$ coded packets are:
\begin{itemize}[leftmargin=*]
\item {\bf Transition to state $A_2$.} Node $A_2$ denotes the reception of the first feedback. The coded packets are retransmitted until the forward link is successful and at least one packet is successfully transmitted. The retransmission is modeled by the self-loop at $A_2$, where
\begin{align}
\matr{P}_{1x}^{\rm D}(1)= z^{\rm RTT}\matr{P}^{\rm RTT}(z\matr{P}_{10}^{\rm C}(1)+z\matr{P}_{11}^{\rm C}(1)z^d \matr{P}^d),\nonumber
\end{align}
where $d=T-{\rm RTT}$ is the residual time for timer expiration upon transmission. Upon the reception of the first feedback, the transition probability matrix $\matr{P}_{10}^{\rm C}(1)$ for the transmission of $M=2$ packets is given by
\begin{align}
\matr{P}_{10}^{\rm C}(1)&=\matr{P}_{10}\matr{P}_{10}+\matr{P}_{10}\matr{P}_{01}+\matr{P}_{01}\matr{P}_{10}, \nonumber
\end{align}
which models the error-free NACK. It combines the different cases such that the feedback is an error-free NACK, i.e., the forward link was bad for both packets and the reverse link was good (first term), or the forward link was bad for either one of the packets only and the reverse link was good (second and third terms). We assume the cumulative feedback is error-free as long as the reverse link is good before the forward transmission is over.

The transition probability matrix $\matr{P}_{11}^{\rm C}(1)$ is given by  
\begin{align}
\matr{P}_{11}^{\rm C}(1)&=\matr{P}_{11}\matr{P}_{11}+\matr{P}_{11}\matr{P}_{10}+\matr{P}_{10}\matr{P}_{11},\nonumber
\end{align}
which models the erroneous NACK feedback. It combines the different cases such that the forward link was bad for both packets and the reverse link (CF) was also bad.

In CF ARQ, unless both packets are successfully acknowledged, we always need retransmissions. Hence, it is suboptimal. Furthermore, the erasure rate of CF ARQ is not the same as the erasure rate of uncoded ARQ. For example, for the case of symmetric memoryless channels, the relationship between the erasure rate $\epsilon_{\rm CF}$ for CF ARQ with $M=2$ packets, and the erasure rate $\epsilon$ of the uncoded ARQ in \cite{AusNos2007} is computed as $\epsilon_{\rm CF} = \sqrt{\epsilon^4+2\epsilon^3 (1-\epsilon)}$. Hence, $\epsilon_{\rm CF}(1)\geq \epsilon^2$.

\item {\bf Transition to state $A_1$.} When the first feedback is received at node $A_2$, if the number of DoFs acknowledged equals $1$, then the system transits to state $A_1$. The matrix 
\begin{align}
\matr{P}_A(1)=\matr{P}_{00}\matr{P}_{10}+\matr{P}_{10}\matr{P}_{00}\nonumber
\end{align} 
denotes the transition probability matrix from $A_2$ to $A_1$. Hence, if the system goes into state $A_1$, the additional number of DoFs required by the receiver is $1$, i.e., only one packet needs to be retransmitted. The packet retransmission at $A_1$ is modeled by the self-loop, where
\begin{align}
\matr{P}_{1x}^{\rm D}(2)= (z\matr{P})^{{\rm RTT}-1}
(z\matr{P}_{10}^{\rm C}(2)+z\matr{P}_{11}^{\rm C}(2)z^{d+1}\matr{P}^{d+1}),\nonumber
\end{align}
where the probability matrices $\matr{P}_{10}^{\rm C}(2)$ and $\matr{P}_{11}^{\rm C}(2)$ model the error-free and the erroneous NACK, respectively. At node $A_1$, as only one packet is retransmitted, the matrices satisfy $\matr{P}_{xy}^{\rm C}(2)=\matr{P}_{xy}$, where $\matr{P}_{xy}$'s, for $x,y\in\{0,1\}$ are the transition probability matrices for the uncoded ARQ in \cite{AusNos2007}. 

\item {\bf Transition to state $O$.} If $N=2$ DoF's are received, the stream can be successfully decoded. If $N=2$ DoF's are acknowledged (with probability $\matr{P}_{00}^{\rm C}(1)=\matr{P}_{00}\matr{P}_{00}$), the system transits to state $O$. 

\item {\bf Transition to state $C_1$.} If $N=2$ DoF's are received, but the feedback is an erroneous ACK (with probability $\matr{P}_{01}^{\rm C}(1)=\matr{P}_{01}\matr{P}_{01}+\matr{P}_{01}\matr{P}_{00}+\matr{P}_{00}\matr{P}_{01}$), then the system transits to $C_1$, where the sender waits till it receives an error-free ACK/NACK, modeled by the self-loop at $C_1$. 

\item {\bf Transition to state $C_2$.}  If $N=2$ DoF's are received, but only one packet is successfully acknowledged and the feedback for the other packet is an erroneous ACK (with probability $\matr{P}_{01}^{\rm C}(2)$), then the system transits to $C_2$, where the sender waits till it receives an error-free ACK/NACK.  
\end{itemize}

Given the transition probabilities, the success and error probability matrices in the reverse channel for CF ARQ are:
\begin{align}
\matr{P}_{x0}^{\rm C}(n)&=\matr{P}_{00}^{\rm C}(n)+\matr{P}_{10}^{\rm C}(n),\nonumber\\
\matr{P}_{x1}^{\rm C}(n)&=\matr{P}_{01}^{\rm C}(n)+\matr{P}_{11}^{\rm C}(n),\,\, n\in\{1,2\}, \nonumber
\end{align}
respectively, where $n-1$ is the number of DoFs acknowledged by the receiver, i.e., $2-(n-1)$ DoFs are needed at the receiver.

\begin{figure*}[t!]
\centering
\includegraphics[width=0.4\textwidth]{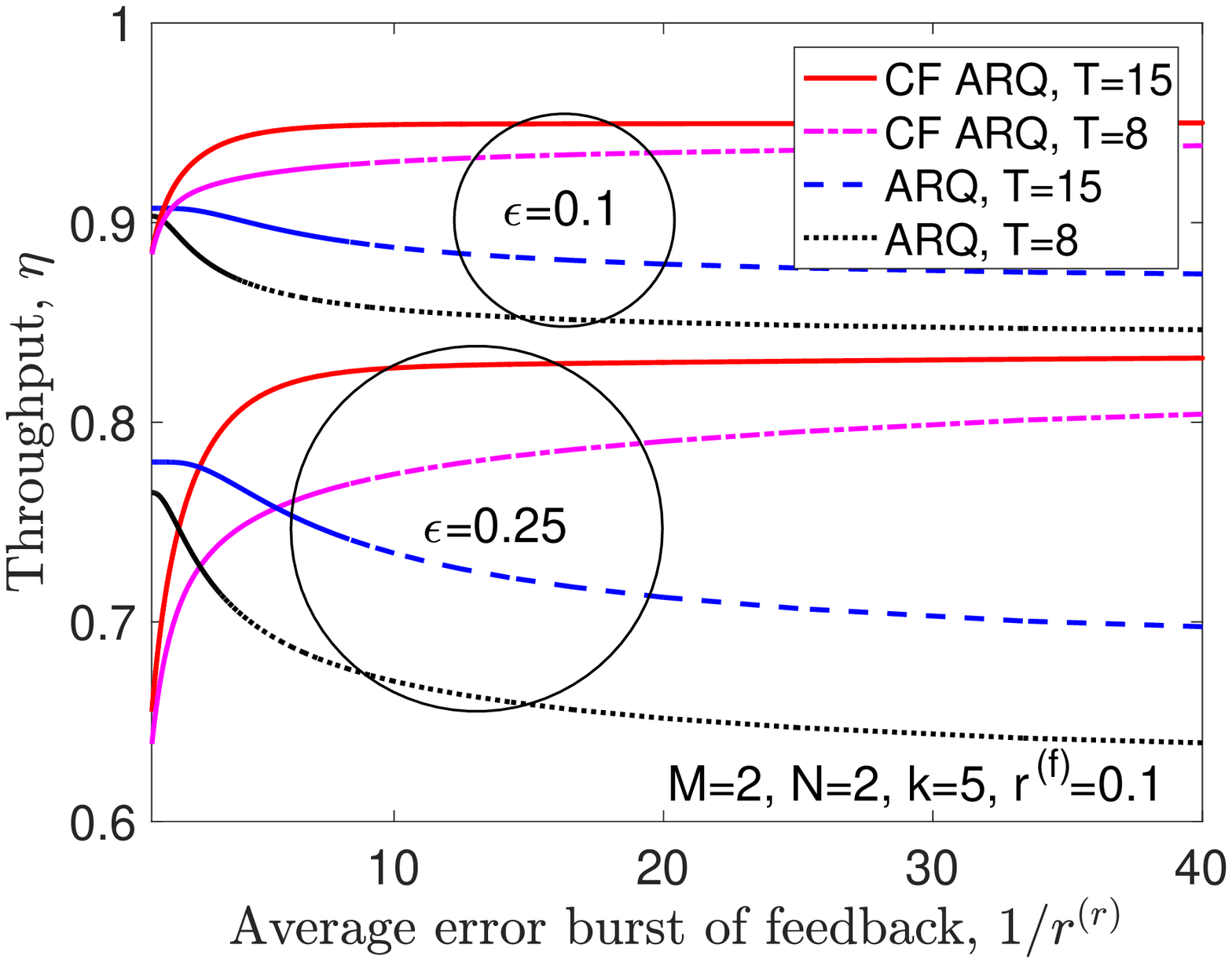}
\hspace{1.5cm}
\includegraphics[width=0.4\textwidth]{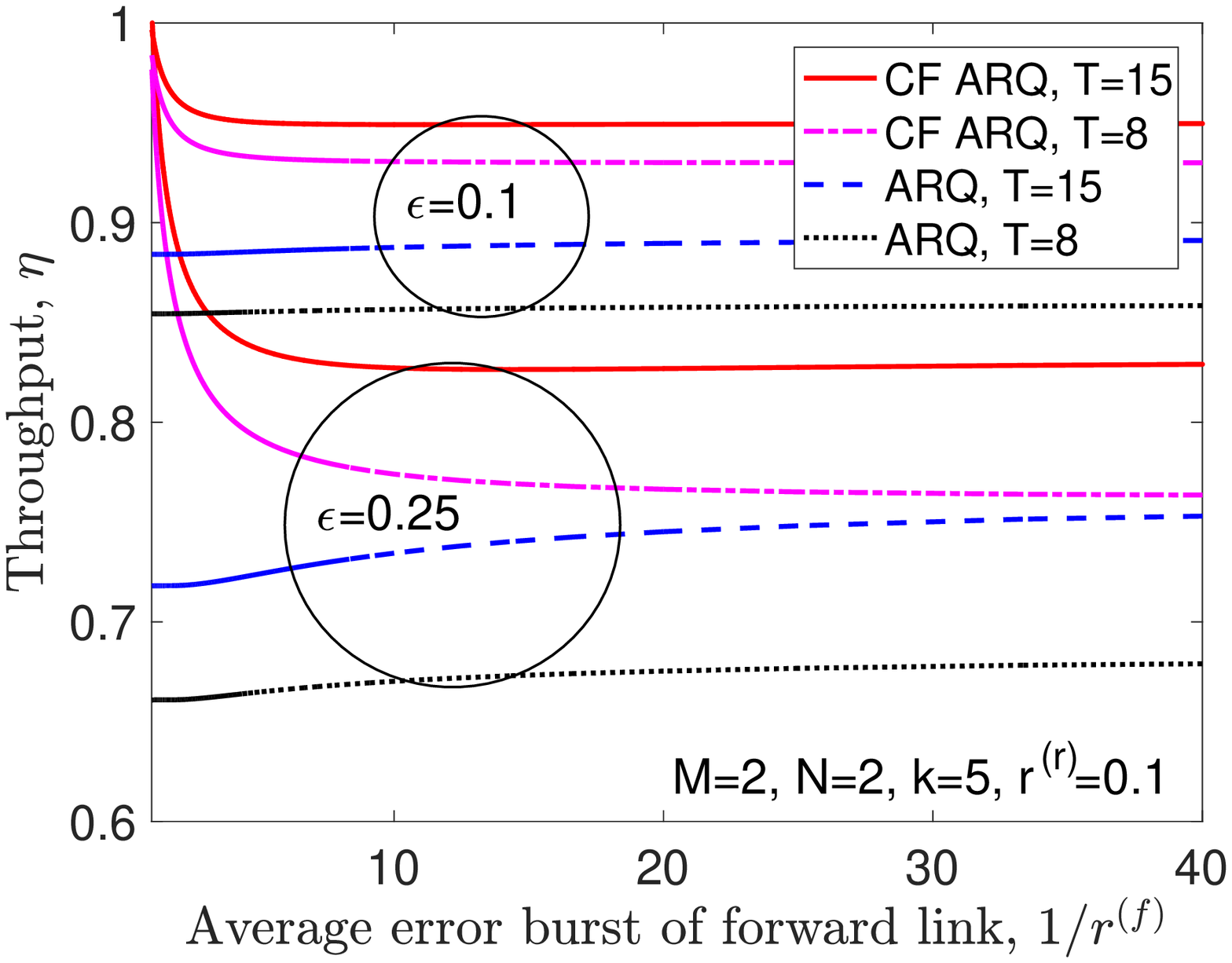}
\includegraphics[width=0.4\textwidth]{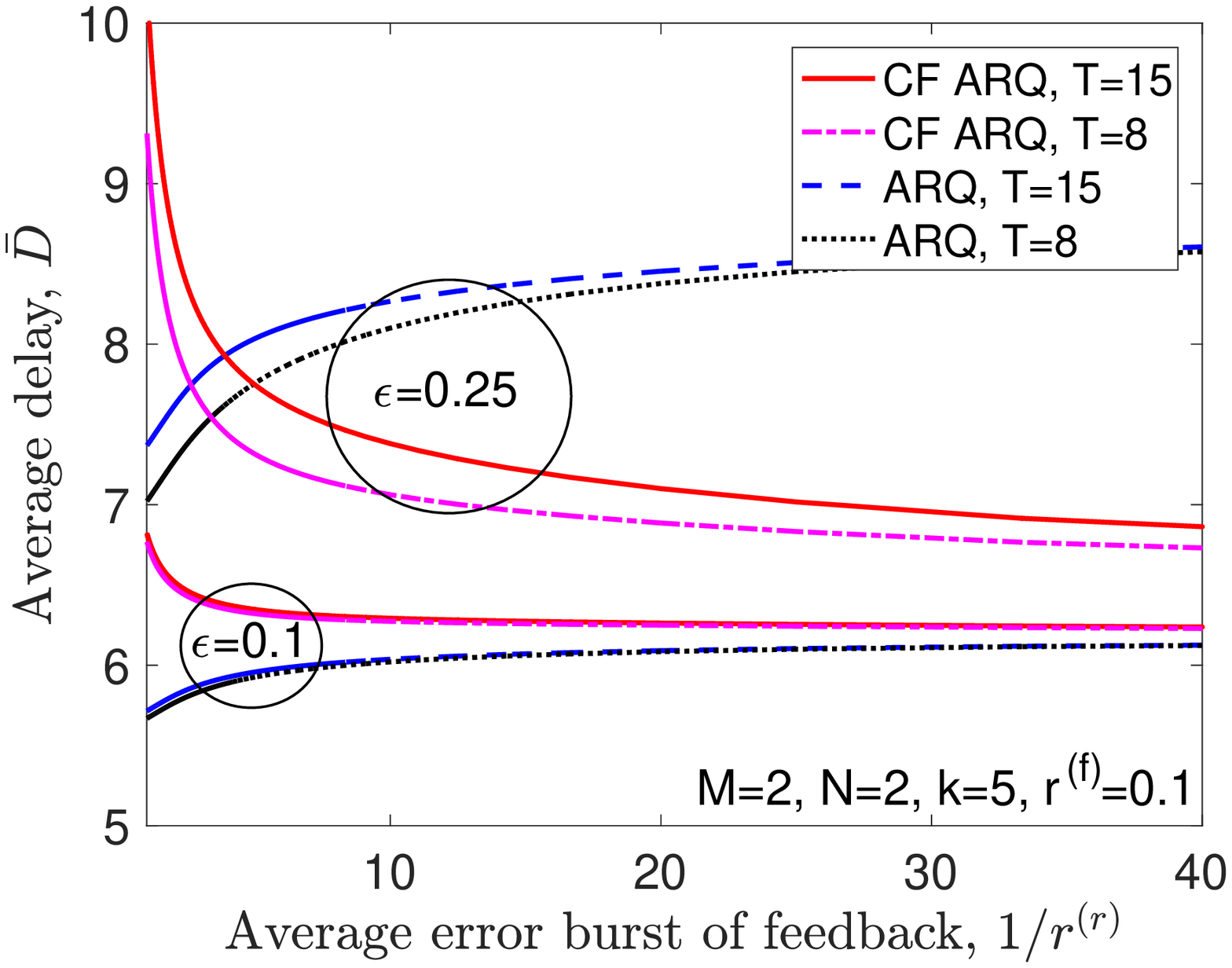}
\hspace{1.5cm}
\includegraphics[width=0.4\textwidth]{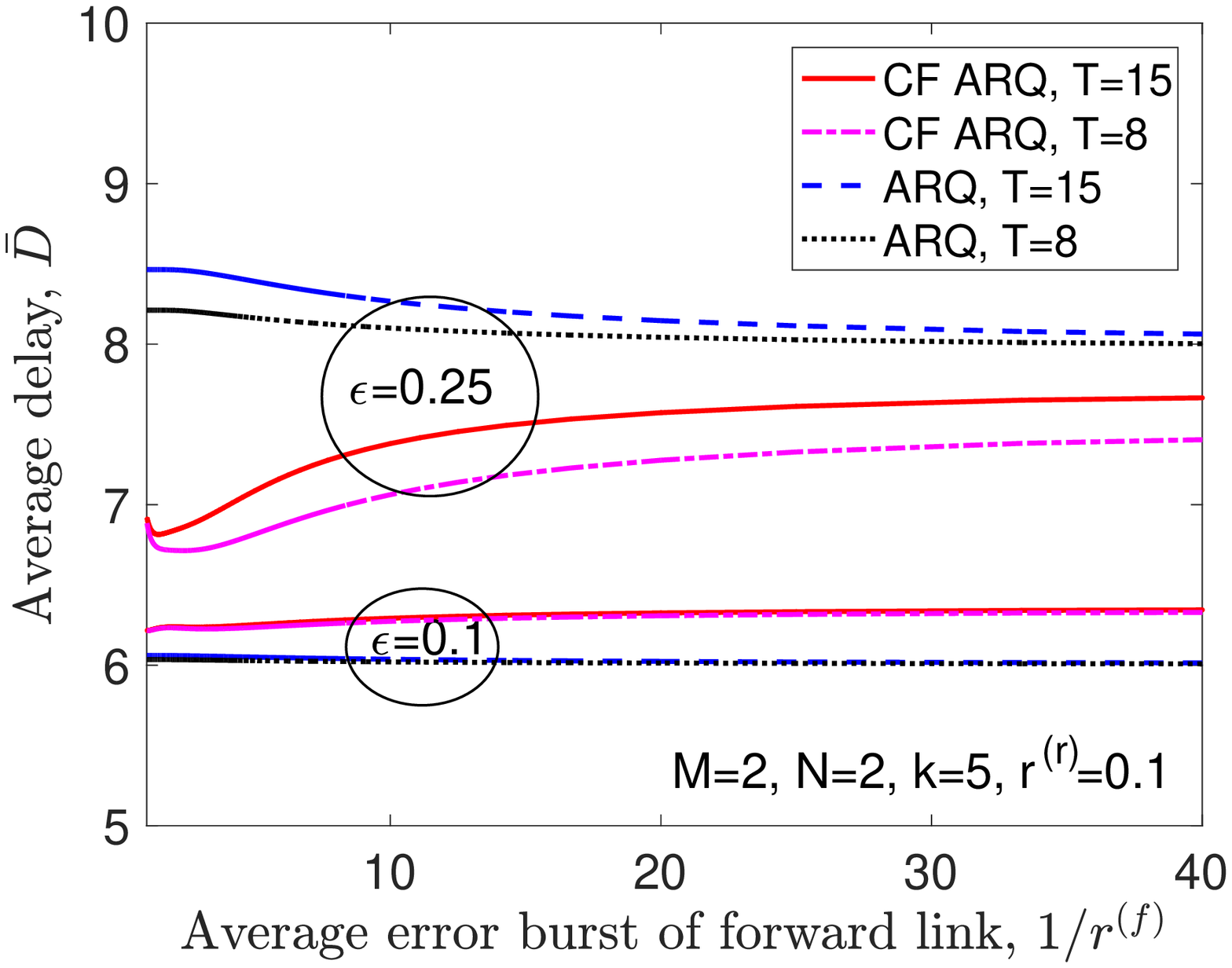}
\caption{\small{Throughput and average delay under different burst models given fixed erasure rates: $\epsilon^{(f)}=\epsilon^{(r)}=\epsilon$.}\label{SymmetricChannels}}
\end{figure*}

The transmission time ($\tau$) is the number of frames being transmitted per a successful frame, and the delay time ($D$) is the time from when a frame is first transmitted to when its ACK is received. Under the given model, both $\tau$ and $D$ are random variables with positive integer outcomes. The matrix gain of the graph in Fig. \ref{SR_ARQ_CF} can be calculated. Using the PGF, the average values for $\tau$ and $D$ can be calculated and the throughput is the reciprocal of $\tau$. We next compute the MGF of the transmission time for $M=2$ packets.

\begin{prop}
\label{CF-ARQthroughput} 
For CF ARQ, the MGF of $\tau$ is given by
\begin{align}
\label{MGFtransmissiontimeCoding}
\matr{\Phi}_{\tau}(z)&=z\matr{P}^{k}\Big[(\matr{I}-\matr{P}_{1x}^{\mathcal{T}}(1))^{-1}\matr{A}_1(z)\nonumber\\
&+\matr{P}_A(1){\prod\nolimits_{i=1}^2}{(\matr{I}-\matr{P}_{1x}^{\mathcal{T}}(i))^{-1}\matr{A}_2(z)}\Big],
\end{align}
where
\begin{align}
\matr{P}_{1x}^{\mathcal{T}}(1)=z(\matr{P}_{10}^{\rm C}(1)+\matr{P}_{11}^{\rm C}(1)(\matr{P}^2)^{d})\matr{P}^{k+1}\nonumber\\
\matr{P}_{1x}^{\mathcal{T}}(2)=z(\matr{P}_{10}^{\rm C}(2)+\matr{P}_{11}^{\rm C}(2)\matr{P}^{d-1})\matr{P}^{k}.\nonumber
\end{align}
and the matrix $\matr{A}_n(z)$ for $n=\{1,2\}$ that gives the gain of the transition from the state $A_{2-(n-1)}$ can be computed as
\begin{align}
\matr{A}_n(z)=\matr{P}_{00}^{\rm C}(n)+\matr{P}_{01}^{\rm C}(n)\left[\sum\nolimits_{i=1}^{d_n}{\matr{P}_{x1}^{\rm C}(n)^{i-1}\matr{P}_{x0}^{\rm C}(n)}\right.\nonumber\\
\left.+\matr{P}_{x1}^{\rm C}(n)^{d_n} (\matr{I}-z\matr{P}_{x1}^{\rm C}(n)^T)^{-1}z\sum\nolimits_{i=0}^{T-1}{\matr{P}_{x1}^{\rm C}(n)^i} \matr{P}_{x0}^{\rm C}(n)\right],\nonumber
\end{align}
where $d_n=d-(n-1)$.  
\end{prop}

The PGF of $\tau$ of CF ARQ for $M=2$ packets is computed as $\Phi_{\tau}(z)=\pi_I\matr{\Phi}_{\tau}(z)\matr{1}/(\pi_I\matr{1})$ using the MGF $\matr{\Phi}_{\tau}(z)$ in (\ref{MGFtransmissiontimeCoding}), where $\matr{1}$ is a column vector of ones, and $\pi_I=\pi \matr{P}_0$ is the probability vector of state $I$. The throughput $\eta$ is the reciprocal of the derivative of $\Phi_{\tau}(z)$ at $z=1$, i.e., $\eta=1/\Phi_{\tau}'(1)$.

We next compute the MGF of the delay for $M=2$ packets.

\begin{prop}\label{CF-ARQdelay}
For CF ARQ, the MGF of the delay is 
\begin{align}
\label{MGFdelayCoding}
\matr{\Phi}_D(z)&=z^{k}\matr{P}^{k}\Big[ (\matr{I}-\matr{P}_{1x}^{\rm D}(1))^{-1}\matr{B}_1(z)\nonumber\\
&+z\matr{P}_A(1)\prod\nolimits_{i=1}^2 (\matr{I}-\matr{P}_{1x}^{\rm D}(i))^{-1}\matr{B}_2(z)\Big],
\end{align}
where $\matr{B}_n(z)$ for $n\in\{1,2\}$ can be computed using relation
\begin{align}
\matr{B}_n(z)=z\matr{P}_{00}^{\rm C}(n)+z\matr{P}_{01}^{\rm C}(n)(\matr{I}-z\matr{P}_{x1}^{\rm C}(n))^{-1}z\matr{P}_{x0}^{\rm C}(n).\nonumber
\end{align}
\end{prop}

The PGF of delay of CF ARQ for $M=2$ coded packets can be computed as $\Phi_D(z)=\pi_I\matr{\Phi}_D(z)\matr{1}/(\pi_I\matr{1})$ using the MGF $\matr{\Phi}_D(z)$ in (\ref{MGFdelayCoding}). Finally, the average delay will be the derivative of $\Phi_D(z)$ at $z=1$, i.e., $\bar{D}=\Phi_D'(1)$.

In Sect. \ref{sims}, we numerically evaluate the throughput and average delay for the GE channel model detailed in Sect. \ref{model} and using the generating functions derived in Sect. \ref{CFcoding}.

\section{Numerical Results}
\label{sims}
We numerically investigate the throughput $\eta$ and average delay $\bar{D}$ of the point-to-point GE channel. Our objective is to understand the impacts of feedback and cumulative feedback (CF) under less reliable and bursty channel conditions.

First assume that the forward and reverse channels have the same erasure rates, i.e., $\epsilon^{(f)}=\epsilon^{(r)}=\epsilon$. In Fig. \ref{SymmetricChannels}, we fix the burst rate of the forward channel, i.e., $r^{(f)}$, and vary the burst rate of the reverse channel, i.e., $r^{(r)}$, and vice versa. As timeout $T$ increases, both the throughput $\eta$ and the average delay $\bar{D}$ increase. As $\epsilon$ increases, it is clear that $\eta$ gets lower and $\bar{D}$ increases both for uncoded ARQ and CF ARQ. Sensitivity of $\eta$ to timeout $T$ also increases under burst errors (low $r$). If the feedback erasures are bursty, $\eta$ decreases with feedback delay. We also observe that $\eta$ of CF ARQ is higher than $\eta$ of uncoded ARQ. The difference becomes significant when the feedback is bursty and $\epsilon$ is high. When $\epsilon$ is small, since the RTTs of CF ARQ and uncoded ARQ are $k+1$ and $k$, respectively, $\bar{D}$ of CF ARQ is higher. However, $\bar{D}$ of CF ARQ is smaller when the feedback is bursty and $\epsilon$ is high. Hence, CF ARQ is more robust to burst errors in the feedback. When there is perfect feedback with $\epsilon^{(r)}=0$, uncoded ARQ has lower $\bar{D}$ and higher $\eta$ than CF ARQ. Throughput and delay performance of both schemes degrade as the feedback loss $\epsilon^{(r)}$ increases. However, CF ARQ outperforms uncoded ARQ under feedback loss.

\begin{figure*}[t!]
\centering
\includegraphics[width=0.4\textwidth]{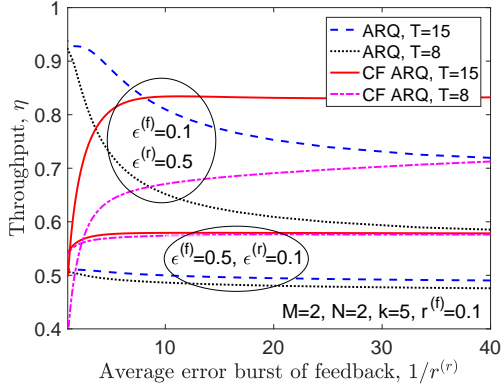}
\hspace{1.5cm}
\includegraphics[width=0.4\textwidth]{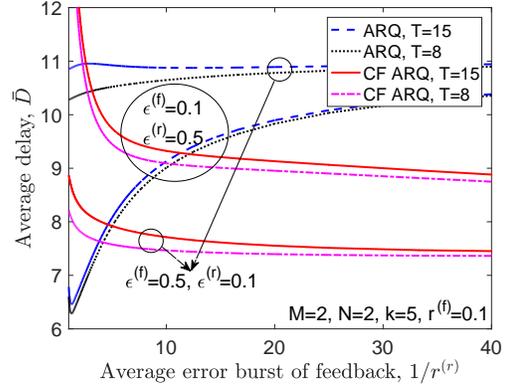}
\caption{\small{Throughput and average delay given $r^{(f)}=0.3$ and different erasure rates: $\epsilon^{(f)}\neq\epsilon^{(r)}$.}\label{AsymmetricChannels}}
\end{figure*}
\begin{figure*}[t!]
\centering
\includegraphics[width=0.4\textwidth]{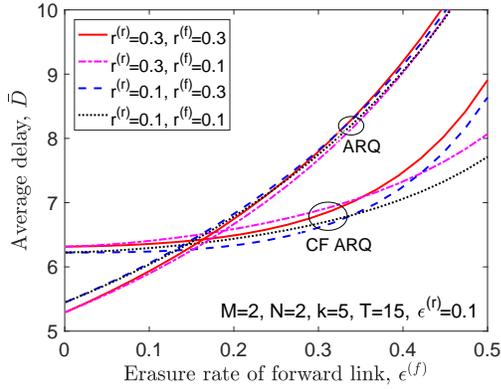}
\hspace{1.5cm}
\includegraphics[width=0.4\textwidth]{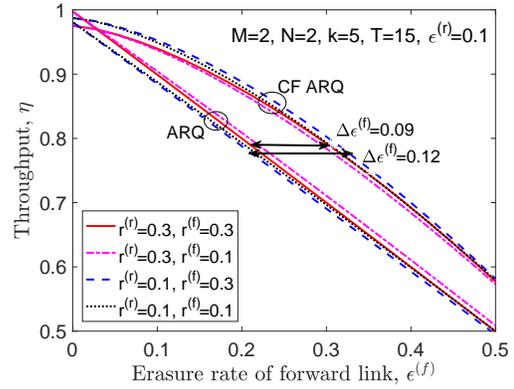}
\caption{\small{Average delay and throughput as function of $\epsilon^{(f)}$, 
under different burst rates.}\label{VaryingErasureRates}}
\end{figure*}

In Fig. \ref{AsymmetricChannels}, we investigate the role of asymmetry such that either the forward channel is more robust to erasures, i.e., $\epsilon^{(f)}=0.1$ and $\epsilon^{(r)}=0.5$, or vice versa. We keep $r^{(f)}$ fixed and increase $r^{(r)}$. When the forward channel is more robust to erasures, $\eta$ is higher both for uncoded ARQ and CF ARQ, and $\bar{D}$ is significantly less compared to the case when the reverse channel is more robust. In this case, CF ARQ provides significantly better throughput than uncoded ARQ. Thus, even if the forward and reverse channels are asymmetric, CF ARQ performs better than uncoded ARQ under bursty feedback.

From Figs. \ref{SymmetricChannels} and \ref{AsymmetricChannels}, we see that erasures of forward channel $\epsilon^{(f)}$ scale the throughput, and feedback erasures $\epsilon^{(r)}$ change the shape of the throughput. Hence, forward erasures $\epsilon^{(f)}$ significantly degrade the throughput of both uncoded ARQ and CF ARQ, and dominate the performance of throughput. Still, CF ARQ throughput gap from uncoded ARQ is higher when $\epsilon^{(r)}$ is high. In terms of delay, CF ARQ is more stable than uncoded ARQ when $\epsilon^{(f)}$ increases, and less stable when $\epsilon^{(r)}$ increases. Still, CF ARQ is more stable when $\epsilon^{(r)}$ is high. 

We next investigate the robustness of CF ARQ to forward erasures. Letting $\epsilon^{(r)}=0.1$, we illustrate $\eta$ and $\bar{D}$ of uncoded ARQ and CF ARQ in Fig. \ref{VaryingErasureRates} for different sets of burst rates. From the plots, we see that CF ARQ provides a higher $\eta$, and even more forward erasures can be compensated (up to $\Delta \epsilon^{(f)}=0.12$) with CF for $M=2$ packets if the feedback channel is more bursty without sacrificing $\bar{D}$. CF ARQ can provide a gain of 18\% in terms of throughput.

Uncoded ARQ is very sensitive to error bursts. The higher the burst rate, the lower its throughput is and the higher its delay is. To compensate for the forward erasures, CF can be used, which can provide significantly less delay, and better throughput. Similarly, when the feedback erasures dominate, performance of CF is much better in terms of throughput, and CF can provide reductions in delay under bursty feedback.

\section{Conclusions}
\label{conc}
We proposed a cumulative feedback-based ARQ with a sliding window of size 2, and computed the MGFs of transmission and delay times. Contrasting its performance with uncoded ARQ, we demonstrated its robustness under burst errors. The following insights should enable more robust design for packet erasure channels with imperfect and bursty feedback:
\begin{itemize}
\item At high erasures, CF ARQ provides considerably low average delay and high throughput than uncoded ARQ.
\item CF ARQ has benefits under burst errors or higher erasure rates in the reverse channel. It is more predictable across statistics, hence is more stable. This can help design robust systems when feedback is unreliable.
\end{itemize}
 
Incorporating FEC, the transmission rate can be adaptively adjusted with the cumulative feedback for multiple coded packets. Extensions hence include the study of different coded schemes, and the throughput achievable with coding. While the analysis is prohibitively complex for larger window sizes with excessive number of hidden states, the technique can easily be evaluated for general window sizes using a network simulator. This can help understand the scalings between the window size and system parameters. This will pave the way for protocol design for 5G with desirable throughput-delay tradeoffs.

\bibliographystyle{IEEEtran}
\bibliography{Derya}

\end{document}